\documentclass[10pt,letterpaper]{article}

\usepackage[T1]{fontenc}
\usepackage[utf8]{inputenc}
\usepackage{lmodern}
\usepackage{microtype}
\usepackage[margin=0.75in]{geometry}
\usepackage{amsmath,amssymb,amsfonts,amsthm,mathtools,bm}
\usepackage{siunitx}
\usepackage{booktabs,array,graphicx,xcolor}
\usepackage[font=small,labelfont=bf]{caption}
\usepackage{stfloats}
\usepackage{tikz,pgfplots}
\pgfplotsset{compat=1.18}
\usepackage{enumitem}
\usepackage{authblk}
\usepackage[colorlinks=true,linkcolor=blue!60!black,citecolor=blue!60!black,urlcolor=blue!60!black]{hyperref}

\theoremstyle{plain}
\newtheorem{proposition}{Proposition}[section]
\newtheorem{theorem}[proposition]{Theorem}
\newtheorem{lemma}[proposition]{Lemma}

\theoremstyle{definition}

\newtheorem{assumption}[proposition]{Assumption}
\theoremstyle{remark}
\newtheorem{remark}[proposition]{Remark}

\newcommand{\dd}{\mathrm{d}}
\newcommand{\R}{\mathbb{R}}
\newcommand{\order}[1]{\mathcal{O}\!\left(#1\right)}
\newcommand{\abs}[1]{\left\lvert #1\right\rvert}

\DeclareMathOperator{\Ei}{Ei}
\newcommand{\Tm}{T_m}
\newcommand{\Ta}{T_a}
\newcommand{\Qstar}{Q^*}
\newcommand{\rhop}{\rho_p}
\newcommand{\sigSB}{\sigma_{\rm SB}}
\newcommand{\rcrit}{r_0^{\rm crit}}
\newcommand{\Abl}{\mathcal{A}}
\newcommand{\chirad}{\chi_{\rm rad}}
\newcommand{\Eh}{\mathcal{E}_\chi}
\newcommand{\Aheat}{A_{\rm h}}
\newcommand{\Arad}{A_{\rm rad}}

\title{A Threshold Model for Micrometeoroid Atmospheric Entry:\\
Filippov Dynamics, Survival Estimates, and Survivor-Only Inverse Limits}
\author[1]{Md Shahrier Islam Arham}
\author[1]{Prasun Panthi}
\author[1]{Min Heo}
\affil[1]{Wabash College, Crawfordsville, IN 47933}
\date{\today}

\begin{document}
\twocolumn[
\begin{@twocolumnfalse}
\maketitle
\begin{abstract}
Micrometeoroids enter Earth's atmosphere at hypervelocity speeds and experience rapid coupling between drag, heating, radiation, melting, ablation, and deceleration. This paper develops a reduced threshold model for the thermal survival boundary of spherical micrometeoroids. The model uses free molecular drag, an exponential atmosphere, radiative cooling from the full spherical surface, and a surplus-heat ablation rule at the melting temperature. The switching surface $T=\Tm$ is treated as a Filippov/complementarity surface. Sustained melting occurs when the local heating-to-radiation ratio exceeds unity. Under the additional Allen--Eggers assumptions of constant radius, constant entry angle, negligible gravity during the main heating interval, and constant transport coefficients, this threshold yields the classical approximate survival scaling $\rcrit\sim v_0^{-3}$. An exact radius-loss identity is obtained along the prescribed Allen--Eggers trajectory, and a perturbative stability estimate explains when this expression approximates the full reduced model. The inverse problem is formulated through a transfer matrix from pre-atmospheric entry bins to observed survivor bins. Entry bins with zero survival probability lie in the survivor-only null space and require external information for reconstruction. The framework gives a compact analytical description of threshold entry survival and identifies the information lost when only surviving particles are observed.
\end{abstract}
\vspace{0.5em}
\vspace{1.2em}
\end{@twocolumnfalse}
]

\section{Introduction}
\label{sec:intro}

Earth accretes interplanetary dust across a broad range of particle sizes, entry speeds, radiant directions, and compositions. Estimates of the total input rate vary because spacecraft detectors, zodiacal cloud models, atmospheric metal layers, ice and snow collections, and recovered micrometeorites sample different parts of the population \cite{plane2012,carrillo2015,carrillo2016,carrillo2020,vanGinneken2024}. Recovered micrometeorites carry a strong atmospheric-entry filter. Particles that survive can be unmelted, partially melted, or fully melted, while particles that fully ablate are absent from collection data \cite{genge2008,genge2020,wozniakiewicz2024}.

Micrometeoroid heating has been studied for decades. Flynn analysed atmospheric entry heating and showed how entry angle, density, and thermal assumptions affect peak temperature estimates \cite{flynn1989}. Love and Brownlee computed numerical histories for particles from about $10~\mu{\rm m}$ to $1~{\rm mm}$ and speeds from $11.2$ to $72~{\rm km\,s^{-1}}$, including mass loss, cooling, gravity, curvature, and atmospheric structure \cite{love1991}. Bronshten and Ceplecha et al. provide the broader meteor-physics background \cite{bronshten1983,ceplecha1998}. Modern ablation models add composition-dependent physics. CAMOD and CABMOD include free molecular flow, sputtering before melting, evaporation from molten particles, diffusion-limited elemental loss, and impact ionisation \cite{vondrak2008,carrillo2015}. Briani et al. solve coupled motion, energy, and mass equations with collisions, radiation, evaporation, melting, sputtering, and the kinetic energy of ablated material \cite{briani2013}. Fragmentation also affects faint meteors and fragile aggregates \cite{campbellbrown2004}.

The present work focuses on the analytical structure of the entry threshold. A four-state model tracks altitude, downward speed, temperature, and radius. Melting is represented by a sharp threshold. Projected area controls drag and aerodynamic heating. Full spherical area controls radiative cooling. At the melting threshold, incoming heat first balances radiation, and the remaining heat drives ablation. This formulation leads naturally to a Filippov/complementarity description. The same threshold also provides a direct route to the familiar $v_0^{-3}$ survival trend under Allen--Eggers assumptions \cite{allen1958}.

The paper makes six contributions. First, it states a threshold entry model with free molecular drag, spherical radiative geometry, and surplus-heat ablation. Second, it identifies the sliding condition on the melting surface. Third, it proves the Allen--Eggers melting-threshold scaling inside the stated reduced problem and gives the finite-altitude correction. Fourth, it derives a radius-loss identity along the Allen--Eggers trajectory and states a stability estimate for its use in the full reduced model. Fifth, it gives a dimensionally consistent nondimensional form, local elasticities, and global sensitivity indices. Sixth, it states the survivor-only inverse problem as a precise matrix non-identifiability result.

\section{Reduced entry model}
\label{sec:model}

\subsection{State variables and geometry}

The state vector is
\begin{equation}
  y(t)=(z(t),v(t),T(t),r(t))^T,
\end{equation}
where $z$ is altitude, $v>0$ is downward speed along the trajectory, $T$ is particle temperature, and $r$ is particle radius. The flight-path angle $\gamma$ is measured from the horizontal and is fixed in the reduced model. The mass and the two relevant areas are
\begin{equation}
  m=\frac{4}{3}\pi\rhop r^3,
  \qquad
  \Aheat=\pi r^2,
  \qquad
  \Arad=\chirad\pi r^2.
\end{equation}
The factor $\chirad$ records the geometry of the radiating surface. For an isothermal sphere,
\begin{equation}
  \chirad=4.
\end{equation}
This value is used in the numerical figures.

The altitude equation is
\begin{equation}
  \dot z=-v\sin\gamma.
  \label{eq:zdot}
\end{equation}
The atmospheric density is modelled as
\begin{equation}
  \rho(z)=\rho_0 e^{-z/H}.
  \label{eq:atmosphere}
\end{equation}
This exponential form gives closed analytical expressions and captures the leading density increase during descent.

\subsection{Momentum balance}

The baseline drag law is free molecular,
\begin{equation}
  D=\Gamma\rho(z)\Aheat v^2,
  \label{eq:drag_force}
\end{equation}
with dimensionless drag coefficient $\Gamma$. Since $v$ is downward speed, gravity contributes positively to $\dot v$ for a descending particle. The speed equation is
\begin{equation}
  \dot v=g\sin\gamma-\Gamma\frac{\rho(z)\Aheat}{m}v^2.
  \label{eq:vdot}
\end{equation}
During the strongest heating interval, the drag term is usually the dominant term in Eq.~\eqref{eq:vdot}.

\subsection{Thermal balance before melting}

The absorbed kinetic energy flux is
\begin{equation}
  Q_H=\frac{1}{2}\Lambda\rho(z)\Aheat v^3,
  \label{eq:qheat}
\end{equation}
where $\Lambda$ is the heat-transfer efficiency. Radiative cooling is
\begin{equation}
  Q_R=\epsilon\sigSB \Arad\left(T^4-\Ta^4\right).
  \label{eq:qrad}
\end{equation}
For $T<\Tm$, the particle does not ablate and
\begin{equation}
  mC_p\dot T=Q_H-Q_R,
  \qquad
  \dot r=0.
  \label{eq:pre_melt_energy}
\end{equation}

\subsection{Surplus-heat ablation at the melting threshold}

At the melting threshold, heat input first balances radiation. The remaining heat drives ablation. Define the surface heat surplus at $T=\Tm$ by
\begin{equation}
  \Eh(z,v)=\frac{1}{2}\Lambda\rho(z)v^3-
  \chirad\epsilon\sigSB\left(\Tm^4-\Ta^4\right).
  \label{eq:excess}
\end{equation}
The common factor $\pi r^2$ has been divided out. Sustained ablation on $T=\Tm$ occurs when
\begin{equation}
  \Eh(z,v)>0.
  \label{eq:excess_positive}
\end{equation}
In this regime, temperature remains at $\Tm$ and the surplus heat produces mass loss:
\begin{equation}
  -\dot m=\frac{\Aheat\Eh(z,v)}{\Qstar}.
  \label{eq:mdot_surplus}
\end{equation}
Using $\dot m=4\pi\rhop r^2\dot r$ and $\Aheat=\pi r^2$ gives
\begin{equation}
  \dot r=-\frac{\Eh(z,v)}{4\rhop\Qstar},
  \qquad T=\Tm,
  \qquad \Eh>0.
  \label{eq:rdot_surplus}
\end{equation}

\subsection{Representative material parameters}

Table~\ref{tab:materials} lists representative values used in the numerical results.

\begin{table}[t]
\caption{Representative material parameters used in the reduced entry model.}
\label{tab:materials}
\begin{tabular}{lcccc}
Material & $\rhop$ & $C_p$ & $\Tm$ & $\Qstar$ \\
 & kg m$^{-3}$ & J kg$^{-1}$ K$^{-1}$ & K & MJ kg$^{-1}$ \\
\hline
Silicate & 3000 & 1000 & 1650 & 5.0 \\
Iron & 7800 & 450 & 1811 & 6.3 \\
Carbonaceous & 1500 & 1000 & 1500 & 4.0 \\
\end{tabular}
\end{table}

\section{Stopped well-posedness}
\label{sec:wellposed}

Because $m\propto r^3$, the reduced equations become singular at complete ablation. The natural mathematical object is therefore a stopped solution. Let
\begin{equation}
  \Omega_{r_{\min}}=\{(z,v,T,r): z>0,\ v>0,\ T>0,\ r>r_{\min}\},
\end{equation}
where $r_{\min}>0$ is the terminal radius used to mark complete ablation.

It is useful to write the threshold rule with a nonnegative mass-loss multiplier $\mu$. Let
\begin{equation}
  q(y)=Q_H-Q_R.
\end{equation}
The constrained thermal balance is
\begin{equation}
  mC_p\dot T=q(y)-\Qstar\mu,
  \qquad
  \dot m=-\mu,
  \qquad
  \mu\geq0.
  \label{eq:multiplier_balance}
\end{equation}
The temperature constraint is $T\leq \Tm$. On $T=\Tm$, the viability condition is $\dot T\leq0$. The minimal nonnegative choice of $\mu$ that enforces this condition is
\begin{equation}
  \mu=\frac{[q(y)]_+}{\Qstar}.
  \label{eq:mu_selection}
\end{equation}
This is the complementarity form of the surplus-heat rule.

\begin{proposition}[Existence up to a terminal event]
\label{prop:existence}
Assume $\rho(z)$ is smooth, all material coefficients are positive, and the initial condition satisfies $y_0\in\Omega_{r_{\min}}$ with $T_0\leq\Tm$. The hard-threshold model admits a Filippov/complementarity solution that can be continued until the first time that $z=0$, $r=r_{\min}$, $v=0$, or the trajectory leaves $\Omega_{r_{\min}}$. Away from $T=\Tm$, the solution is classical and unique. On $T=\Tm$, the selected boundary motion is given by Eq.~\eqref{eq:mu_selection}.
\end{proposition}

\begin{proof}
Fix a compact set $K\subset\Omega_{r_{\min}}$. On $K\cap\{T<\Tm\}$, the right-hand side is smooth because $r\geq r_{\min}>0$, $v>0$, and all coefficients are smooth. Classical local existence and uniqueness follow from the Picard--Lindelof theorem. On the boundary $T=\Tm$, the admissible velocities are described by the closed convex hull of the limiting vector fields from the non-ablating side and the boundary-ablation side. This set-valued map is locally bounded, upper semicontinuous, and has nonempty compact convex values on $K$ \cite{filippov1988,dibernardo2008}. Filippov's existence theorem therefore gives an absolutely continuous solution.

The multiplier formula \eqref{eq:mu_selection} selects a viable boundary velocity. If $q\leq0$, then $\mu=0$ and $\dot T=q/(mC_p)\leq0$, so the trajectory leaves or stays below the boundary. If $q>0$, then $\mu=q/\Qstar$ and $\dot T=0$, so the trajectory remains tangent to $T=\Tm$ while losing mass. Local boundedness on compact subsets gives continuation until the first terminal event listed in the statement. No continuation through $r=0$ is asserted.
\end{proof}

\section{Filippov/complementarity dynamics at the melting surface}
\label{sec:filippov}

The switching surface is
\begin{equation}
  \Sigma=\{(z,v,T,r): T=\Tm\}.
\end{equation}
Let $h(y)=T-\Tm$. The normal component of the pre-melting vector field is
\begin{equation}
  \dot h^+=\dot T^+=\frac{Q_H-Q_R}{mC_p}.
\end{equation}
On $T=\Tm$, this becomes
\begin{equation}
  \dot h^+=\frac{\pi r^2}{mC_p}\Eh(z,v).
  \label{eq:hplus}
\end{equation}

\begin{lemma}[Boundary selection]
\label{lem:boundary_selection}
On $T=\Tm$, the complementarity rule selects
\begin{equation}
  \mu=\frac{[Q_H-Q_R]_+}{\Qstar}.
\end{equation}
Consequently,
\begin{align}
  \dot T&=0 \quad\text{when } Q_H>Q_R,\\
  \dot T&=\frac{Q_H-Q_R}{mC_p}\leq0
  \quad\text{when } Q_H\leq Q_R.
\end{align}
Since $Q_H-Q_R=\pi r^2\Eh(z,v)$, sustained melting occurs exactly when $\Eh(z,v)>0$.
\end{lemma}

\begin{proof}
If $Q_H\leq Q_R$, then $[Q_H-Q_R]_+=0$ and the multiplier rule gives $\mu=0$. Equation \eqref{eq:multiplier_balance} reduces to $mC_p\dot T=Q_H-Q_R\leq0$. If $Q_H>Q_R$, then $\mu=(Q_H-Q_R)/\Qstar$ and Eq.~\eqref{eq:multiplier_balance} gives
\begin{equation}
  mC_p\dot T=Q_H-Q_R-\Qstar\frac{Q_H-Q_R}{\Qstar}=0.
\end{equation}
This proves the two cases. The identity $Q_H-Q_R=\pi r^2\Eh(z,v)$ follows directly from Eqs.~\eqref{eq:qheat} and \eqref{eq:qrad} evaluated at $T=\Tm$.
\end{proof}

It is convenient to write the threshold ratio as
\begin{equation}
  S_\chi(z,v)=\frac{\frac12\Lambda\rho(z)v^3}{\chirad\epsilon\sigSB(\Tm^4-\Ta^4)}.
\end{equation}
Then $\Eh(z,v)>0$ is equivalent to $S_\chi(z,v)>1$.

\section{Allen--Eggers threshold estimate}
\label{sec:AE}

The reduced model admits a clean analytical threshold estimate under the usual Allen--Eggers assumptions.

\begin{assumption}
\label{ass:AE}
Assume: (i) $\rho(z)=\rho_0e^{-z/H}$; (ii) free molecular drag with constant $\Gamma$; (iii) constant radius during the heating estimate; (iv) constant entry angle $\gamma$; (v) negligible gravity during the main heating interval; (vi) constant $\Gamma$, $\Lambda$, $\epsilon$, $\chirad$, $\Tm$, and $\Ta$; and (vii) entry begins high enough that $u_0=e^{-z_0/H}\approx0$.
\end{assumption}

\begin{theorem}[Exact threshold scaling inside the Allen--Eggers problem]
\label{thm:AE_scaling}
Under Assumption~\ref{ass:AE}, the critical initial radius for first reaching the melting threshold is
\begin{equation}
  \rcrit=
  \frac{9\Gamma H e\,\chirad\epsilon\sigSB(\Tm^4-\Ta^4)}{2\Lambda\rhop\sin\gamma\,v_0^3}.
  \label{eq:rcrit_formula}
\end{equation}
In particular,
\begin{equation}
  \rcrit\propto v_0^{-3}.
\end{equation}
\end{theorem}

\begin{proof}
Under Assumption~\ref{ass:AE}, the equations are
\begin{equation}
  \dot v=-\Gamma\frac{\rho \Aheat}{m}v^2,
  \qquad
  \dot z=-v\sin\gamma,
\end{equation}
with $\Aheat/m=3/(4\rhop r_0)$. Hence
\begin{equation}
  \frac{\dd v}{\dd z}=\frac{\dot v}{\dot z}=\Gamma\frac{\rho \Aheat}{m\sin\gamma}v.
\end{equation}
Set $u=e^{-z/H}$. Since $\dd u/\dd z=-u/H$,
\begin{equation}
  \frac{\dd\ln v}{\dd u}=
  \frac{\dd\ln v/\dd z}{\dd u/\dd z}
  =-\beta,
  \qquad
  \beta=\frac{\Gamma\rho_0\Aheat H}{m\sin\gamma}
  =\frac{3\Gamma\rho_0 H}{4\rhop r_0\sin\gamma}.
\end{equation}
Integrating gives
\begin{equation}
  v(u)=v_0e^{-\beta(u-u_0)}.
\end{equation}
The heating factor is
\begin{equation}
  \rho v^3=\rho_0u\,v_0^3e^{-3\beta(u-u_0)}.
\end{equation}
Differentiating $\ln(\rho v^3)$ with respect to $u$ yields
\begin{equation}
  \frac{\dd}{\dd u}\ln(\rho v^3)=\frac1u-3\beta.
\end{equation}
The peak occurs at $u^*=1/(3\beta)$. Therefore
\begin{equation}
  v^*=v_0e^{-1/3},
  \qquad
  \rho^*=\rho_0u^*=\frac{4\rhop r_0\sin\gamma}{9\Gamma H}.
\end{equation}
At the melting threshold,
\begin{equation}
  \frac12\Lambda\rho^*(v^*)^3=\chirad\epsilon\sigSB(\Tm^4-\Ta^4).
\end{equation}
Substituting $\rho^*$ and $v^*$ gives Eq.~\eqref{eq:rcrit_formula}.
\end{proof}

\begin{remark}[Finite-altitude correction]
\label{rem:finite_u0}
If $u_0=e^{-z_0/H}$ is retained, the same calculation gives
\begin{equation}
  \rho^*(v^*)^3=\frac{\rho_0v_0^3}{3\beta}e^{-1+3\beta u_0},
\end{equation}
and the threshold condition becomes the implicit relation
\begin{equation}
  r_0\exp\!\left(\frac{3bu_0}{r_0}\right)=Cv_0^{-3},
\end{equation}
where
\begin{equation}
  b=\frac{3\Gamma\rho_0H}{4\rhop\sin\gamma},
  \qquad
  C=\frac{9\Gamma H e\,\chirad\epsilon\sigSB(\Tm^4-\Ta^4)}{2\Lambda\rhop\sin\gamma}.
\end{equation}
Thus the pure inverse-cubic law is the high-entry-altitude limit of the finite-altitude relation.
\end{remark}

Figure~\ref{fig:threshold} compares the analytical threshold with the numerical threshold for silicate particles. The analytical curve captures the inverse-cubic trend and the correct ordering in radius-speed space.

\begin{figure}[t]
  \centering
  \includegraphics[width=\columnwidth]{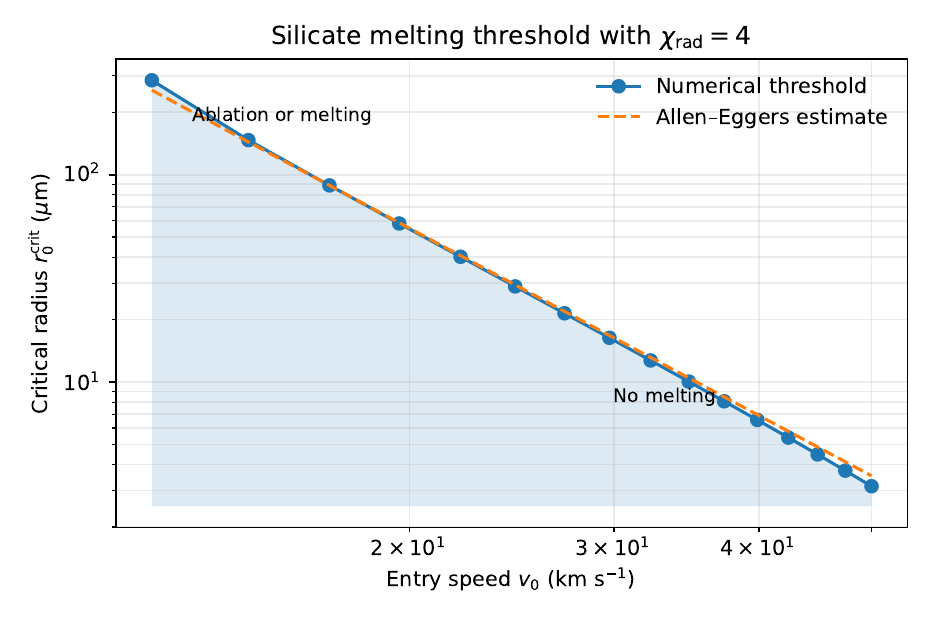}
  \caption{Silicate melting threshold with $\chirad=4$. Circles show the numerical threshold for first reaching $T=\Tm$. The dashed curve is the Allen--Eggers estimate from Theorem~\ref{thm:AE_scaling}. The geometry factor changes the normalisation and leaves the inverse-cubic speed dependence unchanged.}
  \label{fig:threshold}
\end{figure}

\section{Nondimensional form}
\label{sec:nondim}

Set
\begin{equation}
  z=HZ,\qquad v=v_0V,\qquad r=r_0R,\qquad T=\Tm\Theta,
\end{equation}
and use
\begin{equation}
  t=\frac{H}{v_0\sin\gamma}\tau.
\end{equation}
The ambient-temperature ratio is
\begin{equation}
  \Theta_a=\frac{\Ta}{\Tm}.
\end{equation}
Before melting, $R=1$ in the thermal equation if radius is held fixed, and the general pre-melting nondimensional equations are
\begin{align}
  \frac{\dd Z}{\dd\tau}&=-V,\\
  \frac{\dd V}{\dd\tau}&=\Pi_g-\Pi_D e^{-Z}R^{-1}V^2,\\
  \frac{\dd\Theta}{\dd\tau}&=
  \Pi_H e^{-Z}R^{-1}V^3-
  \Pi_R R^{-1}(\Theta^4-\Theta_a^4).
\end{align}
The dimensionless groups are
\begin{align}
  \Pi_g&=\frac{gH}{v_0^2},\\
  \Pi_D&=\frac{3\Gamma\rho_0H}{4\rhop r_0\sin\gamma},\\
  \Pi_H&=\frac{3\Lambda\rho_0Hv_0^2}{8\rhop r_0 C_p\Tm\sin\gamma},\\
  \Pi_R&=\frac{3\chirad\epsilon\sigSB H\Tm^3}{4\rhop r_0 C_p v_0\sin\gamma}.
\end{align}
On the melting surface, define the dimensionless radius-loss rate $\mu_R$ by
\begin{equation}
  \frac{\dd R}{\dd\tau}=-\mu_R.
\end{equation}
This $\mu_R$ is a radius-loss rate and is distinct from the dimensional mass-loss multiplier $\mu$ in Eq.~\eqref{eq:multiplier_balance}. The surplus-heat rule gives
\begin{equation}
  \mu_R=
  \left[
  \Pi_A e^{-Z}V^3-
  \Pi_Q(1-\Theta_a^4)
  \right]_+,
\end{equation}
where
\begin{align}
  \Pi_A&=\frac{\Lambda\rho_0Hv_0^2}{8\rhop\Qstar r_0\sin\gamma},\\
  \Pi_Q&=\frac{\chirad\epsilon\sigSB H\Tm^4}{4\rhop\Qstar r_0 v_0\sin\gamma}.
\end{align}
The local threshold is equivalently
\begin{equation}
  S_\chi=\frac{\Pi_A e^{-Z}V^3}{\Pi_Q(1-\Theta_a^4)}.
\end{equation}
Sustained melting and ablation occur when $S_\chi>1$. This form keeps the geometry factor explicit. It enters the cooling and threshold normalisation, while the cubic dependence on entry speed comes from the kinetic heating term.

\section{Radius-loss identity and perturbative accuracy}
\label{sec:radiusloss}

The Allen--Eggers trajectory also gives an exact radius-loss identity within the prescribed reduced problem.

\begin{theorem}[Radius loss along the prescribed Allen--Eggers trajectory]
\label{thm:radius_loss}
Assume the particle follows the Allen--Eggers velocity profile
\begin{equation}
  v(u)=v_0e^{-\beta(u-u_0)},
  \qquad u=e^{-z/H},
\end{equation}
with fixed radius in the drag law. Assume further that the particle is on the melting surface $T=\Tm$ and that radius loss follows Eq.~\eqref{eq:rdot_surplus}. Then the radius loss between $u_m$ and $u_f$ is
\begin{multline}
\Delta r=
\frac{H}{4\rhop\Qstar\sin\gamma}
\int_{u_m}^{u_f}
\bigg[
\frac12\Lambda\rho_0v_0^2e^{-2\beta(u-u_0)}\\
-
\frac{\chirad\epsilon\sigSB(\Tm^4-\Ta^4)}{uv_0}
 e^{\beta(u-u_0)}
\bigg]_+\dd u.
\label{eq:dr_integral}
\end{multline}
If the bracket is positive on the full interval, then
\begin{align}
\Delta r&=
\frac{H}{4\rhop\Qstar\sin\gamma}
\Bigg[
\frac{\Lambda\rho_0v_0^2}{4\beta}e^{2\beta u_0}
\Big(e^{-2\beta u_m}-e^{-2\beta u_f}\Big)
\nonumber\\
&\quad-
\frac{\chirad\epsilon\sigSB(\Tm^4-\Ta^4)}{v_0}e^{-\beta u_0}
\Big(\Ei(\beta u_f)-\Ei(\beta u_m)\Big)
\Bigg].
\label{eq:dr_closed}
\end{align}
\end{theorem}

\begin{proof}
Along the melting surface,
\begin{multline}
  \dot r=-\frac{1}{4\rhop\Qstar}
  \bigg[
  \frac12\Lambda\rho_0u\,v_0^3e^{-3\beta(u-u_0)}\\
  -\chirad\epsilon\sigSB(\Tm^4-\Ta^4)
  \bigg]_+.
\end{multline}
Since $u=e^{-z/H}$ and $\dot z=-v\sin\gamma$,
\begin{equation}
  \dot u=\frac{uv\sin\gamma}{H},
  \qquad
  \dd t=\frac{H}{uv\sin\gamma}\dd u.
\end{equation}
Therefore
\begin{equation}
  \Delta r=-\int \dot r\,\dd t
  =\frac{H}{4\rhop\Qstar\sin\gamma}
  \int_{u_m}^{u_f}\frac{[\cdots]_+}{uv(u)}\dd u.
\end{equation}
Substituting $v(u)=v_0e^{-\beta(u-u_0)}$ gives Eq.~\eqref{eq:dr_integral}. If the positive part is active throughout the interval, the first term integrates to an exponential and the second to an exponential integral, which yields Eq.~\eqref{eq:dr_closed}.
\end{proof}

\begin{proposition}[Perturbative accuracy in the reduced model]
\label{prop:stability}
Let $\Delta r_{\rm AE}$ be the radius loss computed from Eq.~\eqref{eq:dr_integral} and let $\Delta r_{\rm full}$ be the radius loss in the full reduced model. Consider a compact heating interval $I$ on which $r\ge r_{\min}>0$, $v\ge v_{\min}>0$, the coefficients remain bounded, and the positive-part integrand is locally Lipschitz. Define the fractional Allen--Eggers radius loss by
\begin{equation}
  \chi=\frac{\Delta r_{\rm AE}}{r_0}
\end{equation}
and the gravity-to-drag ratio on $I$ by
\begin{equation}
  \mu_g=\sup_I
  \frac{g\sin\gamma}{\Gamma\rho\Aheat v^2/m}.
\end{equation}
If $\chi<1$ and both trajectories are compared on the same compact interval $I$, then there exists a constant $C_I$ depending only on the compact-domain bounds and Lipschitz constants such that
\begin{equation}
  \abs{\Delta r_{\rm full}-\Delta r_{\rm AE}}
  \le C_I r_0\chi(\chi+\mu_g).
  \label{eq:stability_bound}
\end{equation}
\end{proposition}

\begin{proof}
Write the full reduced model and the frozen-radius Allen--Eggers model as two nonautonomous systems driven by the same altitude variable on the common interval $I$. Because $r(t)=r_0+\order{\chi r_0}$ on $I$, the drag coefficient $\Gamma\rho\Aheat/m$ differs from its frozen-radius value by $\order{\chi}$. The explicit gravity term contributes an additional perturbation of size $\order{\mu_g}$ relative to drag. Standard variation-of-constants estimates for locally Lipschitz systems then give
\begin{equation}
  \sup_I \abs{v_{\rm full}-v_{\rm AE}}
  \le C_1(\chi+\mu_g),
\end{equation}
for a constant $C_1$ depending only on bounds on $I$. Since the surplus integrand in Eq.~\eqref{eq:dr_integral} is locally Lipschitz in $v$, $r$, and $u$, the induced error in the radius-loss integrand is also $\order{\chi+\mu_g}$. Multiplying by the interval length and by the prefactor in Eq.~\eqref{eq:dr_integral} yields an $\order{\chi(\chi+\mu_g)}$ correction to the radius loss. This gives Eq.~\eqref{eq:stability_bound} after absorbing all compact-domain bounds into a single constant $C_I$.
\end{proof}

Figure~\ref{fig:radiusloss} compares Eq.~\eqref{eq:dr_integral} with direct numerical integration. The estimate captures the ordering of radius loss across the sweep. At larger losses it overestimates the numerical result because radius evolution and deceleration modify the heating history relative to the frozen-radius Allen--Eggers trajectory.

\begin{figure}[t]
  \centering
  \includegraphics[width=\columnwidth]{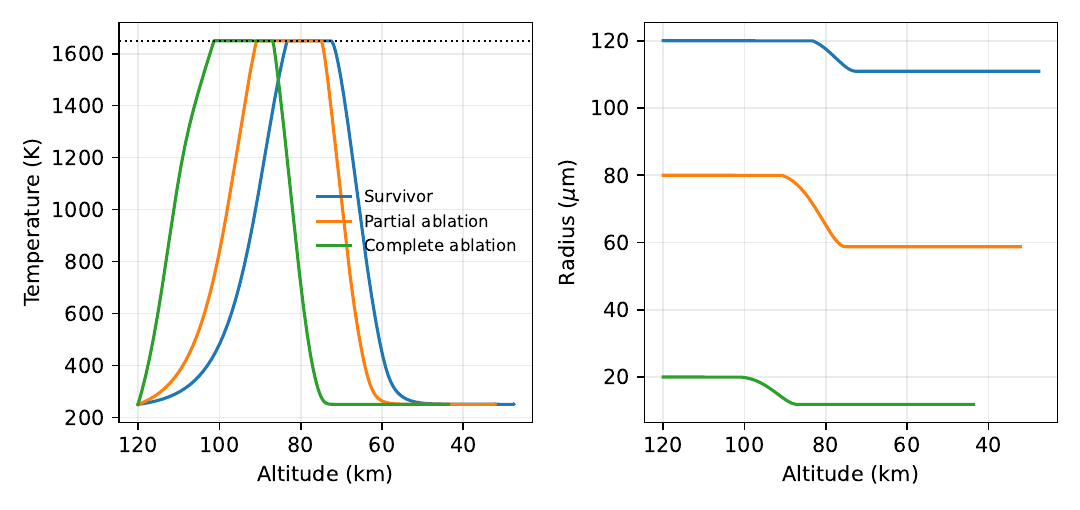}
  \caption{Representative thermal and radius histories for three silicate trajectories. The full spherical radiative area lowers the temperature rise and shifts the onset of sustained ablation to larger radii than the projected-area convention.}
  \label{fig:trajectories}
\end{figure}

\begin{figure}[t]
  \centering
  \includegraphics[width=\columnwidth]{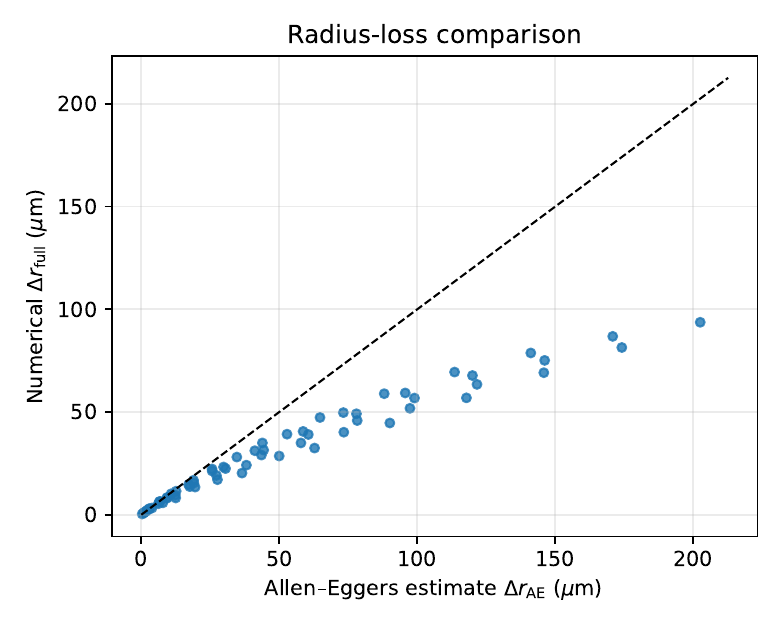}
  \caption{Semi-analytical radius loss from Eq.~\eqref{eq:dr_integral} compared with direct numerical integration for melting silicate trajectories. The diagonal gives exact agreement. The Allen--Eggers estimate reproduces the ranking of radius loss and overestimates the largest losses because the frozen-radius trajectory remains hotter than the evolving full trajectory.}
  \label{fig:radiusloss}
\end{figure}

\section{Sensitivity and fate maps}
\label{sec:sensitivity}

The analytical threshold formula immediately gives local elasticities. From Eq.~\eqref{eq:rcrit_formula},
\begin{align}
  E_{v_0}&=-3,
  & E_{\Lambda}&=-1,
  & E_{\Gamma}&=1,\\
  E_{\rhop}&=-1,
  & E_{\sin\gamma}&=-1,
  & E_{\chirad}&=1.
\end{align}
The elasticity with respect to $\Tm$ is
\begin{equation}
  E_{\Tm}=\frac{4\Tm^4}{\Tm^4-\Ta^4}.
\end{equation}
Figure~\ref{fig:sensitivity} compares these local elasticities, including the radiative-geometry factor, with first-order and total Sobol indices computed from the analytical threshold over finite parameter ranges \cite{saltelli2008}. The two measures answer different questions. Local elasticities describe infinitesimal changes near a baseline point. Sobol indices measure variance contributions across a prescribed input ensemble.

\begin{figure}[t]
  \centering
  \includegraphics[width=\columnwidth]{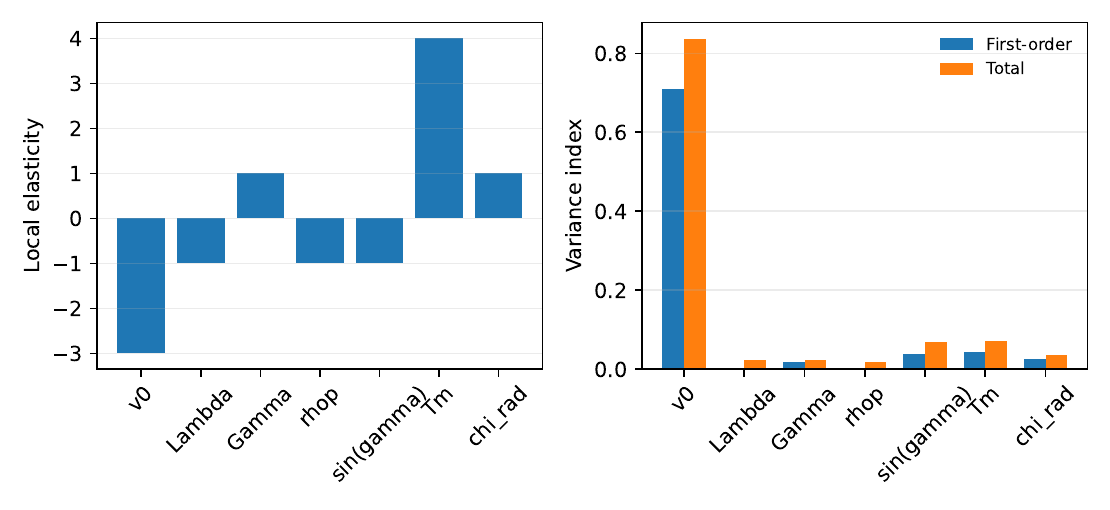}
  \caption{Left: local elasticities of the analytical threshold, including $\chirad$. Right: first-order and total Sobol indices for the same threshold over finite parameter ranges. Entry speed dominates the variance budget because it enters the threshold as $v_0^{-3}$.}
  \label{fig:sensitivity}
\end{figure}

Figures~\ref{fig:fate_silicate} and \ref{fig:fate_iron} show the fate maps for silicate and iron particles. The regimes are unmelted survival, melted survival, partial ablation, and complete ablation. The dashed curve is the Allen--Eggers threshold estimate. The full spherical radiative area shifts the threshold upward while preserving the ordering: low-speed trajectories survive more easily, and increasing speed drives particles toward ablation.

\begin{figure}[t]
  \centering
  \includegraphics[width=\columnwidth]{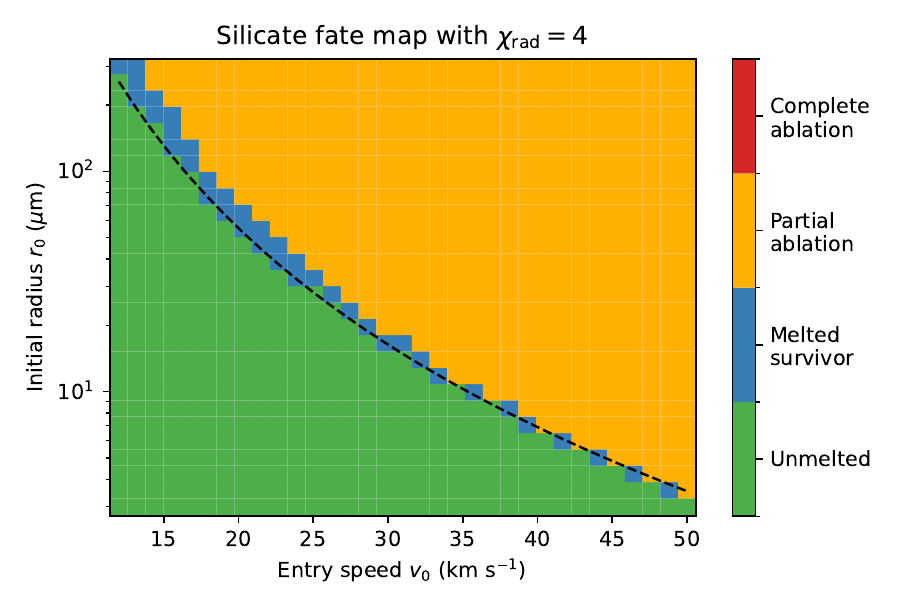}
  \caption{Silicate fate map with $\chirad=4$. The dashed curve is the Allen--Eggers threshold estimate.}
  \label{fig:fate_silicate}
\end{figure}

\begin{figure}[t]
  \centering
  \includegraphics[width=\columnwidth]{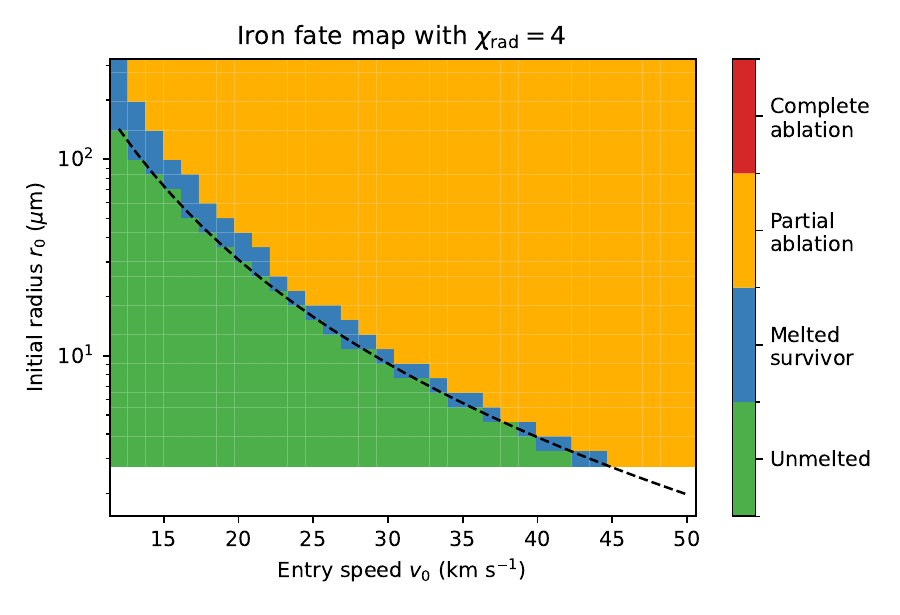}
  \caption{Iron fate map with $\chirad=4$. The larger density pushes the threshold to smaller radii, while the same speed ordering remains visible.}
  \label{fig:fate_iron}
\end{figure}

\section{Benchmark checks against classical entry behaviour}
\label{sec:benchmark}

The reduced model is designed to capture the leading thermal-survival partition in radius-speed space. It is therefore compared at the level of regime ordering rather than detailed chemical ablation. Classical micrometeoroid entry calculations show that survival is favoured at low entry speed, that increasing speed drives particles toward melting and ablation, and that detailed chemical models require additional physics such as sputtering, elemental evaporation, diffusion limits, ionisation, and fragmentation \cite{love1991,vondrak2008,carrillo2015,briani2013,campbellbrown2004}. The threshold model reproduces this qualitative structure. With the spherical radiation convention $\chirad=4$, the Allen--Eggers threshold retains the inverse-cubic speed dependence, while the normalisation shifts according to the emitting-area factor. The fate maps show low-speed survival, transition through melting and partial ablation, and complete ablation in the high-speed region. Table~\ref{tab:benchmark} summarises these checks.

\begin{table*}[t]
\centering
\caption{Benchmark checks against classical micrometeoroid-entry behaviour.}
\label{tab:benchmark}
\begin{tabular}{p{0.20\textwidth}p{0.34\textwidth}p{0.34\textwidth}}
\toprule
Check & Classical entry behaviour & Reduced-model behaviour \\
\midrule
Speed trend & Low entry speeds favour survival; increasing speed drives particles toward melting and ablation. & Fate maps show low-speed survival and transition to melting, partial ablation, and complete ablation as $v_0$ increases. \\
Threshold slope & Classical heating arguments predict a strong inverse dependence on entry speed. & The Allen--Eggers reduction gives $r_0^{\rm crit}\propto v_0^{-3}$. \\
Radiative geometry & A spherical isothermal particle radiates over its full surface. & The model uses $A_{\rm rad}=\chi_{\rm rad}\pi r^2$ with $\chi_{\rm rad}=4$. \\
Model scope & Detailed models include sputtering, elemental evaporation, diffusion limits, ionisation, fragmentation, and composition-dependent chemistry. & The present model isolates thermal threshold survival, radius loss, and survivor-only inverse limits. \\
\bottomrule
\end{tabular}
\end{table*}

\section{Survivor-only inverse problem}
\label{sec:inverse}

Let $x\in\R^{N_e}$ denote the pre-atmospheric entry distribution over a finite set of entry bins. Let $K_{\rm full}\in\R^{(N_s+1)\times N_e}$ be the transfer matrix that maps each entry bin to survivor bins plus one ablation row. The survivor-only observation model is
\begin{equation}
  y_{\rm surv}=K_{\rm surv}x,
\end{equation}
where $K_{\rm surv}\in\R^{N_s\times N_e}$ is obtained from $K_{\rm full}$ by removing the ablation row.

\begin{proposition}[Survivor-only non-identifiability]
\label{prop:inverse}
If column $i$ of $K_{\rm surv}$ is zero, then the entry-bin amplitude $x_i$ is not identifiable from survivor-only data.
\end{proposition}

\begin{proof}
If the $i$th column is zero, then $K_{\rm surv}e_i=0$. For any scalar $\alpha$, 
\begin{equation}
  K_{\rm surv}(x+\alpha e_i)=K_{\rm surv}x+\alpha K_{\rm surv}e_i=K_{\rm surv}x.
\end{equation}
Hence the data do not determine the component $x_i$.
\end{proof}

The deterministic reduced model defines an ideal zero-survival set
\begin{equation}
  \Abl=\{i: p_{\rm surv}(i)=0\}.
\end{equation}
For a sampled transfer matrix, an observed zero column does not by itself prove that the true survival probability is exactly zero. If zero survivors are observed in $n$ independent trials, then a one-sided $(1-\alpha)$ upper confidence bound is
\begin{equation}
  p_{\rm surv}\le 1-\alpha^{1/n}.
\end{equation}
This distinction matters when the matrix is estimated by Monte Carlo rather than by a deterministic sweep.

Figure~\ref{fig:inverse} shows two diagnostics. The left panel gives the survival probability over a coarse input grid. The right panel shows the singular spectrum of the corresponding survivor matrix. Zero-survival regions create structural null directions. The decaying nonzero singular values quantify additional inversion difficulty on the surviving subspace.

\begin{figure*}[t]
  \centering
  \includegraphics[width=0.88\textwidth]{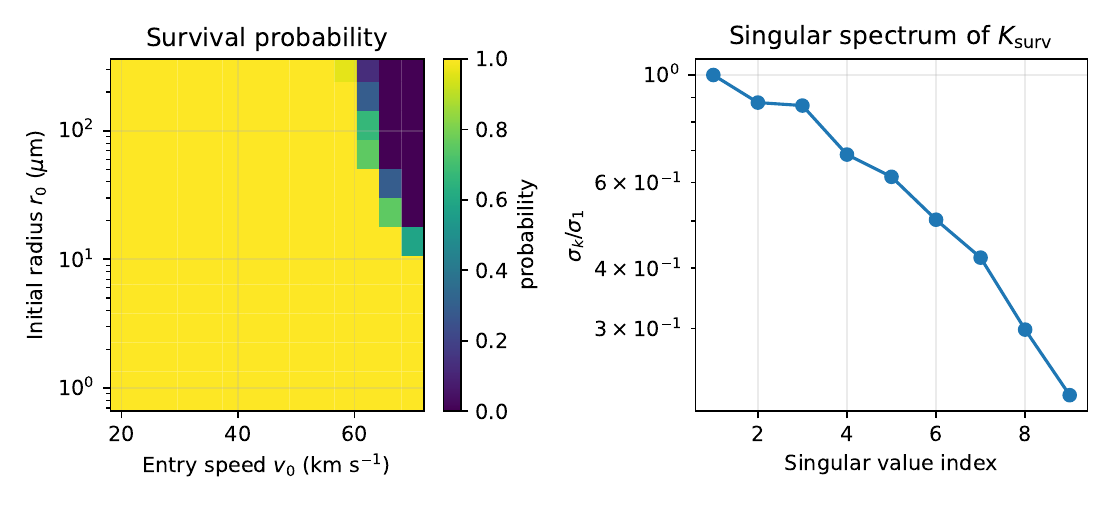}
  \caption{Inverse diagnostics for the survivor-only transfer matrix. Left: survival probability across the entry grid. Right: normalised singular values of $K_{\rm surv}$. Zero-survival regions generate structural null directions, while the remaining singular spectrum measures conditioning on the survivor subspace.}
  \label{fig:inverse}
\end{figure*}

\section{Numerical implementation}
\label{sec:numerics}

All plotted trajectories were generated from the reduced model with $\chirad=4$. The numerical sweeps step downward in altitude from $z_0=120~\mathrm{km}$ using the reduced equations with a fixed altitude increment and event checks for ground arrival, complete ablation, or loss of positivity. When stepping in altitude, time derivatives are converted by
\begin{equation}
  \frac{\dd}{\dd z}=\frac{\dd/\dd t}{-v\sin\gamma}.
\end{equation}
The threshold curve in Fig.~\ref{fig:threshold} is found by a binary search in the initial radius for each speed. The fate maps classify each trajectory according to whether it remains unmelted, survives after melting with little loss, survives after partial ablation, or fully ablates. The inverse diagnostics are built from a coarse entry grid with small random perturbations inside each cell.

\section{Model scope}
\label{sec:scope}

The model is a threshold model. It is designed to expose the geometry of the survival boundary, not to replace full chemical ablation codes. The main assumptions are: spherical particles; a single internal temperature; exponential atmosphere; constant transport coefficients; free molecular drag; no fragmentation; and a melting threshold with surplus-heat ablation. These assumptions make the threshold, scaling law, radius-loss identity, and survivor-only null space explicit. They also limit the model to problems where leading-order survival structure is the main target. Detailed chemistry, differential elemental loss, sputtering, phase transitions beyond the melting threshold, and non-spherical effects remain the domain of more elaborate ablation models such as CAMOD and CABMOD \cite{vondrak2008,carrillo2015,briani2013}.

\section{Conclusion}
\label{sec:conclusion}

This paper develops a reduced threshold model for micrometeoroid atmospheric entry. Free molecular drag and projected-area heating are combined with full-sphere radiative cooling and a surplus-heat ablation rule on the melting surface. The threshold is naturally written as a Filippov/complementarity problem. Within the Allen--Eggers reduction, the critical radius satisfies an exact inverse-cubic speed law with a geometry-dependent normalisation. Along the prescribed Allen--Eggers trajectory, the radius-loss formula is exact. In the full reduced model, that formula remains accurate when radius evolution and gravity are perturbative on the main heating interval. The survivor-only inverse problem has a clear structural null space: zero-survival entry bins cannot be reconstructed without external information. The geometry factor $\chirad$ changes the threshold normalisation, but it does not change the threshold mechanism, the inverse-cubic trend, or the survivor-only inverse limit. These features make the model useful as an analytical threshold description of atmospheric entry survival.

\bibliographystyle{unsrt}
\bibliography{refs}

\end{document}